\documentclass{bmcart}

\usepackage{amsthm,amsmath,amsfonts}
\RequirePackage[authoryear]{natbib}
\RequirePackage{hyperref}
\usepackage[utf8]{inputenc} 
\usepackage{lineno}
\usepackage{graphicx}
\usepackage{booktabs}
\usepackage[skip=2pt, font=footnotesize]{caption}
\usepackage{subcaption}
\usepackage{tikz}
\usetikzlibrary{backgrounds} 
\usetikzlibrary{arrows}

\newtheorem{definition}{Definition}
\newtheorem{theorem}{Theorem}[section]
\newtheorem{lemma}[theorem]{Lemma}

\newcommand{\G}{\mathcal{G}}
\newcommand{\hg}{\mathcal{H}}
\newcommand{\V}{\mathcal{V}}
\newcommand{\E}{\mathcal{E}}
\newcommand{\set}{\mathcal{S}}
\newcommand{\T}{\mathcal{T}}
\newcommand{\R}{\mathbb{R}}

\newcommand{\cut}{\mathrm{cut}}

\newcommand{\st}{\quad\mathrm{s.t.}\,\,}

\startlocaldefs
\endlocaldefs

\begin{document}

\begin{frontmatter}

\begin{fmbox}
\dochead{Research}

\title{Hypergraph cuts with edge-dependent vertex weights}

\author[
  addressref={aff1},                   
  corref={aff1},                       
  email={yz126@rice.edu}               
]{\inits{Y.}\fnm{Yu} \snm{Zhu}}
\author[
  addressref={aff1},
  email={segarra@rice.edu}
]{\inits{S.}\fnm{Santiago} \snm{Segarra}}

\address[id=aff1]{
  \orgdiv{Department of Electrical and Computer Engineering},  
  \orgname{Rice University},                                   
  \street{6100 Main Street},
  \postcode{77005},
  \city{Houston},                                              
  \cny{USA}                                                    
}

\end{fmbox}

\begin{abstractbox}

\begin{abstract} 
We develop a framework for incorporating edge-dependent vertex weights (EDVWs) into the hypergraph minimum $s$-$t$ cut problem.
These weights are able to reflect different importance of vertices within a hyperedge, thus leading to better characterized cut properties.
More precisely, we introduce a new class of hyperedge splitting functions that we call EDVWs-based, where the penalty of splitting a hyperedge depends only on the sum of EDVWs associated with the vertices on each side of the split.
Moreover, we provide a way to construct submodular EDVWs-based splitting functions and prove that a hypergraph equipped with such splitting functions can be reduced to a graph sharing the same cut properties.
In this case, the hypergraph minimum $s$-$t$ cut problem can be solved using well-developed solutions to the graph minimum $s$-$t$ cut problem.
In addition, we show that an existing sparsification technique can be easily extended to our case and makes the reduced graph smaller and sparser, thus further accelerating the algorithms applied to the reduced graph.
Numerical experiments using real-world data demonstrate the effectiveness of our proposed EDVWs-based splitting functions in comparison with the all-or-nothing splitting function and cardinality-based splitting functions commonly adopted in existing work.
\end{abstract}


\end{abstractbox}

\end{frontmatter}

\section{Introduction}\label{s:intro}

The \emph{graph} minimum $s$-$t$ cut problem, or equivalently the maximum $s$-$t$ flow problem, is a fundamental problem in network science.
A cut is a bipartition of the graph vertices and its weight is computed by summing the weights of the edges crossing the cut.
The problem aims to find the minimum weight cut that disconnects the source vertex $s$ from the sink vertex $t$.
It has various applications such as bipartite matching~\citep{cherkassky1998augment, lovasz2009matching}, network reliability~\citep{colbourn1991combinatorial, ramanathan1987counting}, distributed computing~\citep{bokhari1987assignment}, image segmentation~\citep{boykov2006graph, boykov2004experimental} and very large scale integration (VLSI) circuit design~\citep{leighton1999multicommodity, li1995linear}.
Classical solutions to this problem include ``augmenting paths"-based algorithms~\citep{ford1956maximal} and ``push-relabel" style algorithms~\citep{goldberg1988new}, to name a few.

A natural extension is the \emph{hypergraph} minimum $s$-$t$ cut problem. 
Graphs are limited to modeling pairwise relations, while hypergraphs generalize the notion of an edge to a hyperedge, which can represent higher-order interactions connecting more than two vertices.
Many practical problems can be better modeled by hypergraphs~\citep{papa2007hypergraph, schaub2021signal}.
For instance, in the columnwise decomposition of a sparse matrix for parallel sparse-matrix vector multiplication, hypergraphs provide a more accurate representation for the communication volume requirement than graphs, where vertices and hyperedges are respectively used to model columns of the sparse matrix and the non-zero pattern of each row~\citep{catalyurek1999hypergraph}.
In image segmentation, hypergraphs are leveraged to describe higher-order relations among superpixels~\citep{ding2008image, kim2011higher}.
In VLSI circuit design, vertices and hyperedges respectively represent gates and signal nets~\citep{karypis1999multilevel}.

The weight of a hypergraph cut is defined as the sum of splitting penalties associated with every hyperedge.
Different from the graph case, there may exist multiple ways to split a hyperedge.
Consequently, for each hyperedge $e$, we consider a splitting function $w_e:2^e\to\R_{\geq 0}$ that assigns a penalty to every possible cut of $e$ where $2^e$ denotes the power set of $e$.
For any $\set\subseteq e$, $w_e(\set)$ indicates the penalty of partitioning $e$ into $\set$ and $e\setminus\set$~\citep{li2017inhomogeneous, veldt2020hypergraph}.
Existing works mainly adopt two kinds of splitting functions.
One is the so-called all-or-nothing splitting function in which an identical penalty is charged if the hyperedge is cut no matter how it is cut~\citep{hein2013total}.
It is a straightforward extension of the graph case since an edge in a graph is associated with only one non-zero splitting penalty.
Another slightly more general type is the class of cardinality-based splitting functions where the splitting penalty $w_e(\set)$ depends only on the number of vertices placed into $\set$~\citep{veldt2020hypergraph, zhou2006learning}.

There are two major approaches for solving the minimum $s$-$t$ cut problem in hypergraphs.
One is to adopt submodular splitting functions for all hyperedges (cf.~\eqref{e:subf} for the definition of submodular functions), then the hypergraph minimum $s$-$t$ cut problem can be solved using submodular function minimizers~\citep{li2018submodular, veldt2020hypergraph}.
Another more efficient approach is to reduce the hypergraph to a graph that has the same cut properties and then leverage existing solutions to the graph minimum $s$-$t$ cut problem~\citep{ihler1993modeling, lawler1973cutsets, li2017inhomogeneous, veldt2020hypergraph}.
The reduction is generally implemented by expanding every hyperedge into a small graph possibly with additional auxiliary vertices and then concatenating these small graphs to form the final graph.
It has been proved in~\citep{veldt2020hypergraph} that, for cardinality-based splitting functions, the hypergraph cut problem is reducible to a graph cut problem if and only if the splitting functions are submodular.
Moreover,~\cite{veldt2020hypergraph} proposes a graph reduction method for an arbitrary submodular cardinality-based splitting function where a hyperedge $e$ is expanded into a graph that has up to $O(|e|)$ auxiliary vertices and $O(|e|^2)$ edges in the worst case.
This may result in large and dense graphs thus affecting the efficiency of algorithms applied to the reduced graph.
To tackle this problem, sparsification techniques have been developed which try to approximate the hypergraph cut using a sparse graph with fewer auxiliary vertices~\citep{bansal2019new, benczur1996approximating, benson2020augmented, chekuri2018minimum, kogan2015sketching}.
A follow-up paper~\citep{benson2020augmented} proposes a sparsification method for approximating hypergraph cuts defined by submodular cardinality-based splitting functions.
The proposed method reduces the number of auxiliary vertices and the number of edges needed to expand a hyperedge $e$ to $O(\epsilon^{-1}\log |e|)$ and $O(\epsilon^{-1}|e|\log |e|)$ respectively, where $\epsilon$ is the approximation tolerance parameter. 

A disadvantage of the all-or-nothing splitting function as well as cardinality-based ones is that they treat all the vertices in a hyperedge equally while in practice these vertices might contribute differently to the hyperedge.
Such information can be captured by edge-dependent vertex weights (EDVWs):
Every vertex $v$ is associated with a weight $\gamma_e(v)$ for each incident hyperedge $e$ that reflects the contribution of $v$ to $e$~\citep{chitra2019random}.
The hypergraph model with EDVWs is very relevant in practice.
For example, an e-commerce system can be modeled as a hypergraph with EDVWs where vertices and hyperedges respectively correspond to users and products, and EDVWs represent the quantity of a product bought by a user~\citep{li2018tail}.
EDVWs can also be used to model the author positions in a co-authored manuscript~\citep{chitra2019random}, the probability of a pixel belonging to a segment in image segmentation~\citep{ding2010interactive}, and the relevance of a word to a document in text mining~\citep{hayashi2020hypergraph, zhu2021co}, to name a few.

\textbf{Contributions}
In this paper, we propose a new class of splitting functions that we call EDVWs-based.
In an EDVWs-based splitting function, the splitting penalty $w_e(\set)$ depends only on the sum of EDVWs in $\set$, namely $\sum_{v\in\set}\gamma_e(v)$.
Hence, we can write $w_e(\set)=g_e(\sum_{v\in\set}\gamma_e(v))$ for some continuous function $g_e$.
We prove that $w_e$ is submodular if $g_e$ is concave.
The submodularity is necessary for graph reducibility.
We study the EDVWs-based counterparts of four cardinality-based splitting functions in existing work and show that they are graph reducible.
Moreover, we prove that any EDVWs-based splitting function with a concave $g_e$ is graph reducible and provide a way for such a reduction.
We also show that the sparsification technique proposed in~\citep{benson2020augmented} can be easily adapted to the EDVWs-based case. 
The size and the density of the reduced graph depend on both the shape of $g_e$ and the EDVWs' values.
In a nutshell, our paper provides a framework to study hypergraph cut problems incorporating EDVWs and generalizes the results presented in~\citep{benson2020augmented, veldt2020hypergraph} from cardinality-based splitting functions to EDVWs-based ones.

\textbf{Paper outline}
The rest of this paper is structured as follows.
Preliminary concepts and related work about graph and hypergraph cut problems are reviewed in Section~\ref{s:prep}.
The main theoretical results are presented in Section~\ref{s:prop}, where the hypergraph model with EDVWs is introduced in Section~\ref{ss:edvw_model}, the proposed EDVWs-based splitting functions are studied in Section~\ref{ss:edvw_sf}, the graph reducibility results are stated in Section~\ref{ss:edvw_hg2g}, and the sparsification technique is discussed in Section~\ref{ss:sparsify}.
The numerical results shown in Section~\ref{s:exp} validate the effectiveness of introducing EDVWs into hypergraph cuts.
Closing remarks are included in Section~\ref{s:end}. 
 
\section{Preliminaries and related work}\label{s:prep}

\subsection{Graph cuts}\label{ss:g_cut}

Let $\G=(\V,\E,\mathbf{W})$ denote a weighted and possibly directed graph where $\V$ is the vertex set, $\E$ is the edge set, and $\mathbf{W}$ is the weighted adjacency matrix whose entry $W_{uv}$ denotes the weight of the edge from $u$ to $v$.
A cut is a partition of the vertex set $\V$ into two disjoint, non-empty subsets denoted by $\set$ and its complement $\V\setminus\set$.
The weight of the cut is defined as 
\begin{equation}\label{e:g_cut}
\textstyle \cut_{\G}(\set) = \sum_{u\in\set, v\in\V\setminus\set} W_{uv}.
\end{equation}

Given two vertices $s,t$ in the graph, the minimum $s$-$t$ cut problem aims to find the minimum weight cut that separates $s$ and $t$. 
Formally, the problem can be written as
\begin{equation}\label{e:g_st}
\textstyle \min_{\emptyset\subset\set\subset\V} \, \cut_{\G}(\set) \st s\in\set, t\in\V\setminus\set.
\end{equation}

The minimum $s$-$t$ cut problem is the dual of the maximum $s$-$t$ flow problem.
There exist a number of algorithms for the min-cut/max-flow problem and a summary can be found in~\citep{goldberg1998recent}.
Moreover, it is established in \citep{orlin2013max} that the min-cut/max-flow problem is solvable in $O(|\V||\E|)$ time. We also notice that a recent paper \citep{chen2022maximum} provides an algorithm that solves this problem in almost-linear time.

\subsection{Hypergraph cuts}\label{ss:hg_cut}

Let $\hg=(\V,\E)$ be a hypergraph where $\V$ and $\E$ respectively denote the vertex set and the hyperedge set.
Unlike the graph case, a hyperedge can connect more than two vertices thus there may exist multiple ways to split a hyperedge.
For each hyperedge $e\in\E$, we introduce a splitting function $w_e:2^e\to\R_{\geq 0}$ that assigns a non-negative penalty to every possible cut of $e$~\citep{li2017inhomogeneous, veldt2020hypergraph}.
The splitting function satisfies $w_e(\emptyset)=w_e(e)=0$, in other words, a penalty of zero is assigned when the hyperedge is not cut.
Moreover, the splitting function is symmetric if it satisfies $w_e(\set)=w_e(e\setminus\set)$ for any $\set\subseteq e$.
The weight of the hypergraph cut induced by $\set\subseteq\V$ is defined as the sum of splitting penalties associated with every hyperedge, i.e.,
\begin{equation}\label{e:hg_cut}
\textstyle \cut_{\hg}(\set) = \sum_{e\in\E} w_e(\set\cap e).
\end{equation}

There are mainly two types of splitting functions in existing work:
(i) An \emph{all-or-nothing} splitting function assigns the same penalty to every possible cut of the hyperedge regardless of how its vertices are separated~\citep{hein2013total}.
More precisely, $w_e(\set)$ is equal to some positive constant, e.g., the hyperedge weight, for all non-empty $\set\subset e$ and $w_e(\set)=0$ if $\set\in\{\emptyset, e\}$.
(ii) A splitting function is \emph{cardinality-based} if $w_e(\set_1)=w_e(\set_2)$ for all $\set_1,\set_2\subseteq e$ whenever $|\set_1|=|\set_2|$~\citep{veldt2020hypergraph}.
In other words, the value of $w_e(\set)$ depends only on the cardinality of $\set$.
Several examples are given in Table~\ref{table:sf}.

Similar to the graph case, the hypergraph minimum $s$-$t$ cut problem is formulated as
\begin{equation}\label{e:hg_st}
\textstyle \min_{\emptyset\subset\set\subset\V} \, \cut_{\hg}(\set) \st s\in\set, t\in\V\setminus\set.
\end{equation}

For a finite set $\set$, a set function $F:2^\set\to\R$ is called \emph{submodular} if 
\begin{equation}\label{e:subf}
F(\set_1\cup\{v\}) - F(\set_1) \geq F(\set_2\cup\{v\}) - F(\set_2)
\end{equation}
for every $\set_1\subseteq\set_2\subset\set$ and every $v\in\set\setminus\set_2$~\citep{bach2013learning}.
If the splitting function $w_e$ associated with every hyperedge $e\in\E$ is submodular, the resulting hypergraph cut in the form of a sum of submodular functions is also submodular. 
In this case, the hypergraph minimum $s$-$t$ cut problem can be solved using general submodular function minimizers~\citep{grotschel1981ellipsoid, iwata2003faster, iwata2001combinatorial, iwata2009simple, orlin2009faster, schrijver2000combinatorial} or minimizers for decomposable submodular functions~\citep{ene2017decomposable, kolmogorov2012minimizing, li2018revisiting}. 
The all-or-nothing splitting function and the cardinality-based ones listed in Table~\ref{table:sf} are all submodular.

\subsection{Graph reducibility}\label{ss:hg2g}

Since algorithms for the graph minimum $s$-$t$ cut problem are more efficient than algorithms for general submodular function minimization, another way of solving hypergraph cut problems is to reduce the hypergraph to a graph that shares the same or similar cut properties~\citep{ihler1993modeling, lawler1973cutsets, li2017inhomogeneous}.
The reduction is generally accomplished via hyperedge expansions.
Recently, a generalized hyperedge expansion has been formulated in~\citep{veldt2020hypergraph}, which projects a hyperedge onto a graph allowing directed edges and additional vertices.
The formal definition is given below.

\begin{definition}[Gadget splitting function~\citep{veldt2020hypergraph}]
A gadget associated with a hyperedge $e$ is a weighted and possibly directed graph $\G_e=(\V',\E')$ with vertex set $\V'=e\cup\hat{\V}$ where $\hat{\V}$ is a set of auxiliary vertices.
The corresponding gadget splitting function $\hat{w}_e:2^e\to\R_{\geq 0}$ is defined as
\begin{equation}\label{e:gadget_w}
\textstyle	\hat{w}_e(\set) = \min_{\T\subseteq\V', \T\cap e=\set} \, \cut_{\G_e}(\T).
\end{equation}
\end{definition}

A hyperedge splitting function is \emph{graph reducible} if it is identical to some gadget splitting function.  
A hypergraph cut function defined as~\eqref{e:hg_cut} or the hypergraph minimum $s$-$t$ cut problem is graph reducible if all its hyperedge splitting functions are graph reducible.
It has been proved in~\citep{veldt2020hypergraph} that every gadget splitting function $\hat{w}_e$ defined as~\eqref{e:gadget_w} is submodular.
Hence, if a hyperedge splitting function is graph reducible, it must be submodular.

In the following, we give several examples of splitting functions that have been shown to be graph reducible in existing works.
The all-or-nothing splitting function is graph reducible and can be constructed from the Lawler gadget described as follows.

\emph{Lawler gadget}~\citep{lawler1973cutsets}. 
The Lawler gadget replaces a hyperedge $e$ with a digraph defined on the vertex set $\V'=e\cup\{e',e''\}$ where $e',e''$ are two auxiliary vertices.
For each $v\in e$, add a directed edge of weight infinity from $v$ to $e'$ and a directed edge of weight infinity from $e''$ to $v$.
Finally, add a directed edge of weight equal to the hyperedge weight from $e'$ to $e''$.

The cardinality-based splitting functions listed in Table~\ref{table:sf} are all graph reducible and correspond to the following gadgets, respectively.

\emph{Clique gadget}~\citep{agarwal2006higher}. 
This gadget is an (undirected) clique graph with vertex set $\V'=e$.
For every $u,v\in e$, add an edge of weight $1$ between $u$ and $v$.

\emph{Star gadget}~\citep{agarwal2006higher}.
This gadget is an (undirected) star graph with vertex set $\V'=e\cup\{v_e\}$ where $v_e$ is an auxiliary vertex.
For each $v\in e$, add an edge of weight $1$ between $v$ and $v_e$.

\emph{Symmetric cardinality-based gadget}~\citep{veldt2020hypergraph}.
Similar to the Lawler gadget, this gadget is a digraph with vertex set $\V'=e\cup\{e',e''\}$.
For each $v\in e$, add a directed edge of weight $1$ from $v$ to $e'$ and a directed edge of weight $1$ from $e''$ to $v$.
Moreover, add a directed edge of weight $b\in\mathbb{N}$ from $e'$ to $e''$.

\emph{Asymmetric cardinality-based gadget}~\citep{veldt2020hypergraph}.
This gadget is a digraph defined on $\V'=e\cup\{v_e\}$.
For each $v\in e$, add a directed edge of weight $a$ from $v$ to $v_e$ and a directed edge of weight $b$ from $v_e$ to $v$.

\begin{table}
\fontsize{6}{7}\normalfont
\centering
\caption{Examples of cardinality-based and EDVWs-based splitting functions and their corresponding gadgets where $\set$ is a subset of the hyperedge $e$ and $a,b$ are positive constants.}
\label{table:sf}
\begin{tabular}{lll}
\toprule
Cardinality-based & EDVWs-based & Corresponding gadget \\
\midrule
$w_e(\set)=|\set|\cdot|e\setminus\set|$       & $w_e(\set)=\gamma_e(\set)\cdot\gamma_e(e\setminus\set)$       & Clique gadget \\
$w_e(\set)=\min\{|\set|,|e\setminus\set|\}$   & $w_e(\set)=\min\{\gamma_e(\set),\gamma_e(e\setminus\set)\}$   & Star gadget \\
$w_e(\set)=\min\{|\set|,|e\setminus\set|,b\}$ & $w_e(\set)=\min\{\gamma_e(\set),\gamma_e(e\setminus\set),b\}$ & Sym. cardinality/EDVWs-based gadget \\
$w_e(\set)=\min\{a|\set|,b|e\setminus\set|\}$ & $w_e(\set)=\min\{a\gamma_e(\set),b\gamma_e(e\setminus\set)\}$ & Asym. cardinality/EDVWs-based gadget \\
\bottomrule
\end{tabular}
\end{table}

\section{Hypergraph cuts with EDVWs}\label{s:prop}

\subsection{The hypergraph model with EDVWs}\label{ss:edvw_model}

In this paper, we consider the hypergraph model with EDVWs as defined next.

\begin{definition}[Hypergraph with EDVWs~\citep{chitra2019random}]
Let $\hg=(\V,\E,\kappa,\{\gamma_e\})$ be a hypergraph with EDVWs where $\V$ and $\E$ respectively denote the vertex set and the hyperedge set.
The function $\kappa:\E\to\R_{+}$ assigns positive weights to hyperedges and those weights reflect the strength of connection. 
Each hyperedge $e\in\E$ is associated with a function $\gamma_e:e\to\R_{+}$ to assign EDVWs.
For convenience, we define $\gamma_e(\set)=\sum_{v\in\set}\gamma_e(v)$ for any $\set\subseteq e$.
\end{definition}

The introduction of EDVWs enables the hypergraph to model the cases when the vertices in the same hyperedge contribute differently to this hyperedge.
For example, in a coauthorship network, every author (vertex) in general has a different degree of contribution to a paper (hyperedge), usually reflected by the order of the authors.
This information is lost in traditional hypergraph models but it can be easily encoded through EDVWs.
In the following, we study how to incorporate EDVWs into hypergraph cut problems.

\subsection{EDVWs-based splitting functions}\label{ss:edvw_sf}

A natural extension from cardinality-based splitting functions to EDVWs-based ones is to make the splitting penalty $w_e(\set)$ dependent only on the sum of EDVWs in $\set$. 

\begin{definition}[EDVWs-based splitting function]
We refer to splitting functions defined in the following way as EDVWs-based splitting functions:
\begin{equation}\label{e:EDVWs_w}
w_e(\set) = g_e(\gamma_e(\set)), \quad \forall\set\subseteq e,      
\end{equation}
where $g_e:[0,\gamma_e(e)]\to\R_{\geq 0}$ satisfies $g_e(0)=g_e(\gamma_e(e))=0$.
\end{definition}

For trivial EDVWs, namely $\gamma_e(v)=1$ for all $v\in e$, we have $\gamma_e(\set)=|\set|$ and the EDVWs-based splitting function reduces to a cardinality-based one.
Actually, $g_e$ can be viewed as a continuous extension of the splitting function $w_e$.
In practice, we can also incorporate the hyperedge weight $\kappa(e)$ into the splitting function such as setting $w_e(\set) = \kappa(e)\cdot g_e(\gamma_e(\set))$.
This does not influence the results presented in this paper.

We are interested in submodular splitting functions which make it possible to leverage existing solvers for submodular function minimization.
Moreover, as mentioned in Section~\ref{ss:hg2g}, submodularity is a necessary condition for a splitting function to be graph reducible~\citep{veldt2020hypergraph}.
In the following Theorem~\ref{thm:concave2subm}, we show that the EDVWs-based splitting function defined as~\eqref{e:EDVWs_w} is submodular if $g_e$ is concave.
We first present several properties of concave functions which will be used later.

\begin{lemma}[Properties of concave functions]\label{lemma:concave}
For a concave function $g$, we have
\begin{enumerate}
\item[(i)] If $b_1\leq b_2$, $a>0$, the inequality $g(b_1+a)-g(b_1)\geq g(b_2+a)-g(b_2)$ holds. 
\item[(ii)] If $b_1<b_2<b_3$, the inequality $g(b_2)\geq\frac{b_3-b_2}{b_3-b_1}g(b_1)+\frac{b_2-b_1}{b_3-b_1}g(b_3)$ holds. 
\item[(iii)] If $b_1<b_2\leq a$ and $g$ is symmetric with respect to $a$, the inequality $g(b_1)\leq g(b_2)$ holds. 
\end{enumerate}
\end{lemma}

\begin{proof}
According to the definition of concave functions, the following inequality holds for any $x$ and $y$ in the domain of a concave function $g$,
\begin{equation*}
g(tx+(1-t)y)\geq tg(x)+(1-t)g(y), \quad\forall t\in[0,1].	
\end{equation*}
To prove Property (i), we set $x=b_1$ and $y=b_2+a$. 
For $t=\frac{b_2-b_1}{b_2-b_1+a}$ and $t=\frac{a}{b_2-b_1+a}$, the above inequality respectively becomes
\begin{align*}
g(b_1+a) &\geq \textstyle \frac{b_2-b_1}{b_2-b_1+a}g(b_1) + \frac{a}{b_2-b_1+a}g(b_2+a), \\
g(b_2) &\geq \textstyle \frac{a}{b_2-b_1+a}g(b_1) + \frac{b_2-b_1}{b_2-b_1+a}g(b_2+a).
\end{align*}
Property (i) can be obtained by respectively adding both sides of these two inequalities together.
Property (ii) can be proved by setting $x=b_1$, $y=b_3$ and $t=\frac{b_3-b_2}{b_3-b_1}$.
Property (iii) can be proved by setting $x=b_1$, $y=2a-b_1$ and $t=\frac{2a-b_1-b_2}{2a-2b_1}$.
\end{proof}

\begin{theorem}\label{thm:concave2subm}
For EDVWs-based splitting functions defined as~\eqref{e:EDVWs_w}, if $g_e$ is a concave function, then $w_e$ is submodular.
If $g_e$ is concave as well as symmetric with respect to $\gamma_e(e)/2$, then $w_e$ is submodular and symmetric.
\end{theorem}

\begin{proof}
For every $\set_1\subseteq\set_2\subset e$ and every $v\in e\setminus\set_2$, set $b_1=\gamma_e(\set_1)$, $b_2=\gamma_e(\set_2)$ and $a=\gamma_e(v)$.
When $g_e$ is concave, it follows from (i) in Lemma~\ref{lemma:concave} that
\begin{equation*}
g_e(\gamma_e(\set_1)+\gamma_e(v)) - g_e(\gamma_e(\set_1)) \geq 
g_e(\gamma_e(\set_2)+\gamma_e(v)) - g_e(\gamma_e(\set_2)).
\end{equation*}
It immediately follows from~\eqref{e:EDVWs_w} that 
\begin{equation*}
w_e(\set_1\cup\{v\}) - w_e(\set_1) \geq w_e(\set_2\cup\{v\}) - w_e(\set_2).
\end{equation*}
Hence, $w_e$ is submodular according to the definition of submodular functions. 
Moreover, it is straightforward to show that $w_e$ satisfies $w_e(\set)=w_e(e\setminus\set)$ for all $\set\subseteq e$ if $g_e$ is symmetric with respect to $\gamma_e(e)/2$.
\end{proof}

\subsection{Graph reducibility of EDVWs-based splitting functions}\label{ss:edvw_hg2g}

We first consider the following concave functions for $g_e$.
\begin{align}
g_e(x) &= x\cdot(\gamma_e(e)-x),                    \label{e:concave1} \\
g_e(x) &= \min\{x, \gamma_e(e)-x\},                 \label{e:concave2} \\
g_e(x) &= \min\{x, \gamma_e(e)-x, b\}, \,\, b>0,    \label{e:concave3} \\
g_e(x) &= \min\{ax, b(\gamma_e(e)-x)\}, \,\, a,b>0. \label{e:concave4}
\end{align}
By substituting these concave functions into~\eqref{e:EDVWs_w}, we obtain the EDVWs-based splitting functions listed in Table~\ref{table:sf}, which are submodular according to Theorem~\ref{thm:concave2subm}. 
The first three are also symmetric while the last one is generally asymmetric unless $a=b$.
For the third one, if the parameter $b$ is set to some value no greater than $\min_{v\in e}\gamma_e(v)$, it reduces to the all-or-nothing splitting function; while if $b\geq\max_{\set\subseteq e}\min\{\gamma_e(\set), \gamma_e(e)-\gamma_e(\set)\}$, it reduces to the second one.
We are going to show that the EDVWs-based splitting functions listed in Table~\ref{table:sf} are also graph reducible by constructing gadgets that generalize the gadgets introduced in Section~\ref{ss:hg2g}.
An illustration of these generalized gadgets is given in Figure~\ref{fig:gadget}.

\begin{table}
\fontsize{6}{7}\normalfont
\centering
\caption{Examples of gadgets $\G_e=(e\cup\hat{\V}, \E')$ incorporating EDVWs.}
\label{table:gadget}
\begin{tabular}{p{3.5cm}p{1.2cm}p{0.8cm}p{1.4cm}p{4.1cm}}
\toprule
Gadget name & Type & $\hat{\V}$ & $|\E'|$ & Description \\
\midrule
(EDVWs-based) Clique gadget & Undirected & $\emptyset$ & $|e|(|e|-1)/2$ & For every $u,v\in e$, add an edge $(u,v)$ of weight $\gamma_e(u)\gamma_e(v)$. \vspace{0.3em} \\ 
(EDVWs-based) Star gadget & Undirected & $\{v_e\}$ & $|e|$ & For every $v\in e$, add an edge $(v,v_e)$ of weight $\gamma_e(v)$. \vspace{0.3em} \\ 
Sym. EDVWs-based gadget & Directed & $\{e',e''\}$ & $2|e|+1$ & For every $v\in e$, add a directed edge of weight $\gamma_e(v)$ from $v$ to $e'$ and a directed edge of weight $\gamma_e(v)$ from $e''$ to $v$. Moreover, add a directed edge of weight $b$ from $e'$ to $e''$. \vspace{0.3em} \\
Asym. EDVWs-based gadget & Directed & $\{v_e\}$ & $2|e|$ & For every $v\in e$, add a directed edge of weight $a\gamma_e(v)$ from $v$ to $v_e$ and a directed edge of weight $b\gamma_e(v)$ from $v_e$ to $v$.\\
\bottomrule
\end{tabular}
\end{table}
 
\begin{figure*}[t!]
\centering
\includegraphics[scale=1]{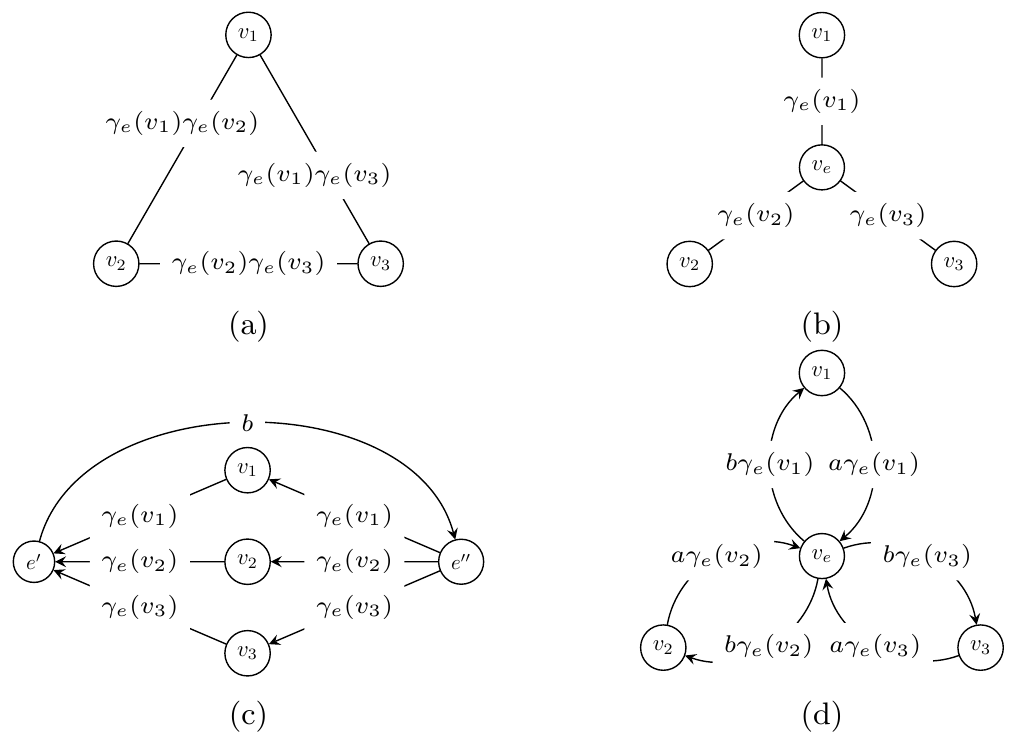}
\captionsetup{margin=0.3cm}
\caption{The illustration of (a) a clique gadget, (b) a star gadget, (c) a symmetric EDVWs-based gadget, and (d) an asymmetric EDVWs-based gadget, for a hyperedge $e=\{v_1,v_2,v_3\}$.}
\label{fig:gadget}
\end{figure*}
 
\begin{theorem}
The EDVWs-based splitting functions listed in Table~\ref{table:sf} are all graph reducible and respectively correspond to the gadgets described in Table~\ref{table:gadget}.
\end{theorem}

\begin{proof}
In the following, we prove these four cases one by one. 

(i) It follows from~\eqref{e:gadget_w} that the splitting function constructed from the clique gadget is
\begin{equation*}
\hat{w}_e(\set) = \cut_{\G_e}(\set) = \sum_{u\in\set, v\in e\setminus\set}\gamma_e(u)\gamma_e(v)=\sum_{u\in\set}\gamma_e(u)\cdot\sum_{v\in e\setminus\set}\gamma_e(v)=\gamma_e(\set)\cdot\gamma_e(e\setminus\set).
\end{equation*}

(ii) According to~\eqref{e:gadget_w}, there are two ways to split $\set$ from $e$ in the star gadget: set $\T=\set$ or $\T=\set\cup\{v_e\}$. The corresponding graph cut weights $\cut_{\G_e}(\T)$ are respectively equal to $\gamma_e(\set)$ and $\gamma_e(e\setminus\set)$. Take the minimum one and the result follows. 

(iii) For the symmetric EDVWs-based gadget, there are four ways to split $\set$ from $e$: set $\T=\set$, $\T=\set\cup\{e',e''\}$, $\T=\set\cup\{e'\}$ or $\T=\set\cup\{e''\}$. They respectively result in $\cut_{\G_e}(\T)$ equal to $\gamma_e(\set)$, $\gamma_e(e\setminus\set)$, $b$ and $\gamma_e(e)$. Notice that the last one is always no less than the first two. The result follows by taking the minimum one. 

(iv) To split $\set$ from $e$ in the asymmetric EDVWs-based gadget, we can set $\T=\set$ or $\T=\set\cup\{v_e\}$ which respectively lead to $\cut_{\G_e}(\T)$ equal to $a\gamma_e(\set)$ and $b\gamma_e(e\setminus\set)$. The proof is completed by taking the minimum one.
\end{proof}

It has been proved in~\citep{veldt2020hypergraph} that all submodular cardinality-based splitting functions are graph reducible. 
More precisely, any submodular cardinality-based splitting function can be constructed from a combination of up to $|e|-1$ different asymmetric cardinality-based gadgets.
For a symmetric submodular cardinality-based splitting function, it can also be constructed from a combination of up to $\lfloor |e|/2 \rfloor$ symmetric cardinality-based gadgets.
In Theorem~\ref{thm:hg2g} below, we show graph reducibility for the proposed submodular EDVWs-based splitting functions.
To this end, we first define two sets as follows.
\begin{align}
& \mathcal{Q}_s=\{\gamma_e(\set)|\set\subseteq e, 0<\gamma_e(\set)\leq\gamma_e(e)/2\}, \\
& \mathcal{Q}_a=\{\gamma_e(\set)|\emptyset\subset\set\subset e\}.
\end{align}
Denote by $\set_i$ the subset of $e$ corresponding to the $i$th smallest element in $\mathcal{Q}_a$.
The set $\mathcal{Q}_s$ is a subset of $\mathcal{Q}_a$ and contains the smallest $|\mathcal{Q}_s|$ elements in $\mathcal{Q}_a$.

\begin{theorem}\label{thm:hg2g}
All EDVWs-based splitting functions defined as~\eqref{e:EDVWs_w} with concave $g_e$ are graph reducible.
A hyperedge paired with such a splitting function can be reduced to a graph which is a combination of at most $|\mathcal{Q}_a|$ asymmetric EDVWs-based gadgets.
If, in addition, $g_e$ is symmetric, the hyperedge can also be reduced to a graph combining at most $|\mathcal{Q}_s|$ symmetric EDVWs-based gadgets.
\end{theorem}

\begin{proof} 
We first consider the case when $g_e$ is concave and symmetric, then study the more general case when $g_e$ is concave and possibly asymmetric.

(i) Consider a hyperedge $e$ and an EDVWs-based splitting function $w_e$ with concave, symmetric $g_e$.
We are going to show that $w_e$ corresponds to some gadget splitting function $\hat{w}_e$ and one way for constructing such a gadget is to combine $r=|\mathcal{Q}_s|$ symmetric EDVWs-based gadgets.
For the $i$th symmetric EDVWs-based gadget, denote the two auxiliary vertices by $e'_i$ and $e''_i$, set the weight of the edge from $e'_i$ to $e''_i$ to $b_i=\gamma_e(\set_i)$, then scale all edge weights by a factor $a_i\geq 0$.
The combined gadget contains $|e|+2|\mathcal{Q}_s|$ vertices and $(2|e|+1)|\mathcal{Q}_s|$ edges.
Its corresponding splitting function can be written as
\begin{equation}\label{e:sym_combined}
\textstyle \hat{w}_e(\set)=\sum_{i=1}^r a_i\cdot\min\{\gamma_e(\set), \gamma_e(e\setminus\set), b_i\}.
\end{equation}
By substituting $\set_i$ into~\eqref{e:sym_combined} we get 
\begin{align*}
\hat{w}_e(\set_i) 
&= \textstyle \sum_{j=1}^r a_j\cdot\min\{b_i, \gamma_e(e)-b_i, b_j\} \\
&= \textstyle \sum_{j=1}^i a_j b_j + \sum_{j=i+1}^r a_j b_i, \quad\forall i=1,\cdots,r,
\end{align*}
which can be condensed in the following matrix form
\begin{equation}\label{e:sym_combined_mat}
\underbrace{\left[\begin{matrix}
b_1 & b_1 & b_1 & b_1 & \cdots & b_1 \\
b_1 & b_2 & b_2 & b_2 & \cdots & b_2 \\
b_1 & b_2 & b_3 & b_3 & \cdots & b_3 \\
b_1 & b_2 & b_3 & b_4 & \cdots & b_4 \\
\vdots & \vdots & \vdots &\vdots & \ddots & \vdots \\
b_1 & b_2 & b_3 & b_4 & \cdots & b_r 
\end{matrix}\right]}_{\mathbf{B}_1}
\left[\begin{matrix}
a_1 \\
a_2 \\
a_3 \\
a_4 \\
\vdots \\
a_r	
\end{matrix}\right] = 
\left[\begin{matrix}
\hat{w}_e(\set_1) \\
\hat{w}_e(\set_2) \\
\hat{w}_e(\set_3) \\
\hat{w}_e(\set_4) \\
\vdots \\
\hat{w}_e(\set_r)
\end{matrix}\right].
\end{equation}
The question left is whether there exist non-negative $a_1,\cdots,a_r$ such that $\hat{w}_e(\set_i)=w_e(\set_i)$ for all $i\in[r]$.
To identify such $a_i$, we replace $\hat{w}_e(\set_i)$ with $w_e(\set_i)$ and invert the system~\eqref{e:sym_combined_mat} as follows
\begin{equation*}
\left[\begin{matrix}
a_1     \\
a_2     \\
\vdots  \\
a_{r-1} \\
a_r	
\end{matrix}\right] = \mathbf{B}_1^{-1}
\left[\begin{matrix}
w_e(\set_1)     \\
w_e(\set_2)     \\
\vdots          \\
w_e(\set_{r-1}) \\
w_e(\set_r)
\end{matrix}\right],
\end{equation*}
where 
\begin{equation*}
\mathbf{B}_1^{-1} = 
\left[\begin{smallmatrix}
\frac{b_2}{(b_2-b_1)b_1} & -\frac{1}{b_2-b_1} & 0 & \cdots & 0 & 0 & 0 \\
-\frac{1}{b_2-b_1} & \frac{b_3-b_1}{(b_3-b_2)(b_2-b_1)} & -\frac{1}{b_3-b_2} & \cdots & 0 & 0 & 0 \\
\vdots & \vdots & \vdots & \ddots & \vdots & \vdots & \vdots\\
0 & 0 & 0 & \cdots & -\frac{1}{b_{r-1}-b_{r-2}} & \frac{b_r-b_{r-2}}{(b_r-b_{r-1})(b_{r-1}-b_{r-2})} & -\frac{1}{b_r-b_{r-1}} \\
0 & 0 & 0 & \cdots & 0 & -\frac{1}{b_r-b_{r-1}} & \frac{1}{b_r-b_{r-1}}
\end{smallmatrix}\right]. 
\end{equation*}
It follows that 
\begin{equation*}
\left[\begin{matrix}
a_1     \\
a_2     \\
\vdots  \\
a_{r-1} \\
a_r	
\end{matrix}\right] = 
\left[\begin{matrix}
\frac{b_2}{(b_2-b_1)b_1} w_e(\set_1) - \frac{1}{b_2-b_1} w_e(\set_2) \\
\frac{b_3-b_1}{(b_3-b_2)(b_2-b_1)} w_e(\set_2) - \frac{1}{b_2-b_1} w_e(\set_1) - \frac{1}{b_3-b_2} w_e(\set_3) \\
\vdots \\
\frac{b_r-b_{r-2}}{(b_r-b_{r-1})(b_{r-1}-b_{r-2})} w_e(\set_{r-1}) - \frac{1}{b_{r-1}-b_{r-2}} w_e(\set_{r-2}) - \frac{1}{b_r-b_{r-1}} w_e(\set_r) \\
\frac{1}{b_r-b_{r-1}} w_e(\set_r) - \frac{1}{b_r-b_{r-1}} w_e(\set_{r-1})
\end{matrix}\right].
\end{equation*}
Since $a_1,\cdots,a_r$ need to be non-negative, we are left to prove the following inequalities:
\begin{align}
w_e(\set_1) &\textstyle \geq \frac{b_1}{b_2} w_e(\set_2), \label{e:ineq_1} \\
w_e(\set_i) &\textstyle \geq \frac{b_{i+1}-b_i}{b_{i+1}-b_{i-1}} w_e(\set_{i-1}) + \frac{b_i-b_{i-1}}{b_{i+1}-b_{i-1}} w_e(\set_{i+1}), \quad\forall i=2,\cdots,r-1, \label{e:ineq_2} \\
w_e(\set_r) &\geq w_e(\set_{r-1}). \label{e:ineq_3}
\end{align}
Moreover, it can be observed from~\eqref{e:sym_combined_mat} that $\hat{w}_e(\set_1)\leq\hat{w}_e(\set_2)\leq\cdots\leq\hat{w}_e(\set_r)$ due to the structure of $\mathbf{B}_1$ and the non-negativity of coefficients $a_i$. Hence, we also need to show that
\begin{equation}\label{e:ineq_4}
w_e(\set_1)\leq w_e(\set_2)\leq\cdots\leq w_e(\set_r). 
\end{equation}
By introducing $b_0=0$, the inequalities~\eqref{e:ineq_1},~\eqref{e:ineq_2} can be rewritten as
\begin{equation}\label{e:ineq_5}
g_e(b_i) \textstyle\geq \frac{b_{i+1}-b_i}{b_{i+1}-b_{i-1}} g_e(b_{i-1}) + \frac{b_i-b_{i-1}}{b_{i+1}-b_{i-1}} g_e(b_{i+1}), \quad\forall i=1,\cdots,r-1.
\end{equation}
It follows immediately from (ii) in Lemma~\ref{lemma:concave}. 
The inequalities~\eqref{e:ineq_4} (including~\eqref{e:ineq_3}) can be rewritten as 
\begin{equation}\label{e:ineq_6}
g_e(b_1)\leq g_e(b_2)\leq\cdots\leq g_e(b_r),
\end{equation}
which can be proved according to (iii) in Lemma~\ref{lemma:concave}.
Notice that, when there exists any equality in \eqref{e:ineq_1}-\eqref{e:ineq_3}, the corresponding $a_i$ equals $0$, which implies that the number of symmetric EDVWs-based gadgets needed to construct the combined gadget can be further reduced.
The equality in~\eqref{e:ineq_2} (or see~\eqref{e:ineq_5}) means that the points at $b_{i-1}$, $b_i$ and $b_{i+1}$ are colinear.

(ii) Next we consider the case when $g_e$ is concave and possibly asymmetric. 
In this case, $w_e$ can be shown to be identical to some gadget splitting function $\hat{w}_e$ and such a gadget can be constructed by combing $r=|\mathcal{Q}_a|$ asymmetric EDVWs-based gadgets.
The $i$th one is paired with parameters $a_i (\gamma_e(e)-b_i)$ and $a_i b_i$ where $a_i\geq 0$ is a scaling parameter.
The combined gadget consists of $|e|+|\mathcal{Q}_a|$ vertices and $2|e|\cdot|\mathcal{Q}_a|$ edges.
Its corresponding splitting function can be written as
\begin{equation}\label{e:asym_combined}
\textstyle \hat{w}_e(\set)=\sum_{i=1}^r a_i\cdot\min\{(\gamma_e(e)-b_i)\gamma_e(\set), b_i\gamma_e(e\setminus\set)\}.
\end{equation}
We set $b_i=\gamma_e(\set_i)$ for all $i\in[r]$.
Then it holds that $b_i+b_{r+1-i}=\gamma_e(e)$.
By substituting $\set_i$ into~\eqref{e:asym_combined} we get
\begin{align*}
\hat{w}_e(\set_i) 
&= \textstyle \sum_{j=1}^r a_j\cdot\min\{(\gamma_e(e)-b_j)b_i, b_j(\gamma_e(e)-b_i)\},   \\
&= \textstyle \sum_{j=1}^i a_jb_jb_{r+1-i} + \sum_{j=i+1}^r a_jb_{r+1-j}b_i, \quad\forall i=1,\cdots,r,
\end{align*}
which can also be written in the following matrix from
\begin{equation}\label{e:asym_combined_mat}
\underbrace{\left[\begin{matrix}
b_1b_r     & b_1b_{r-1} & b_1b_{r-2} & b_1b_{r-3} & \cdots & b_1b_2     & b_1b_1    \\
b_1b_{r-1} & b_2b_{r-1} & b_2b_{r-2} & b_2b_{r-3} & \cdots & b_2b_2     & b_2b_1    \\
b_1b_{r-2} & b_2b_{r-2} & b_3b_{r-2} & b_3b_{r-3} & \cdots & b_3b_2     & b_3b_1    \\
b_1b_{r-3} & b_2b_{r-3} & b_3b_{r-3} & b_4b_{r-3} & \cdots & b_4b_2     & b_4b_1    \\
\vdots     & \vdots     & \vdots     & \vdots     & \ddots & \vdots     & \vdots    \\
b_1b_2     & b_2b_2     & b_3b_2     & b_4b_2     & \cdots & b_{r-1}b_2 & b_{r-1}b_1 \\
b_1b_1     & b_2b_1     & b_3b_1     & b_4b_1     & \cdots & b_{r-1}b_1 & b_rb_1     \\
\end{matrix}\right]}_{\mathbf{B}_2}
\left[\begin{matrix}
a_1 \\
a_2 \\
a_3 \\
a_4 \\
\vdots \\
a_{r-1} \\
a_r
\end{matrix}\right] = 
\left[\begin{matrix}
\hat{w}_e(\set_1) \\
\hat{w}_e(\set_2) \\
\hat{w}_e(\set_3) \\
\hat{w}_e(\set_4) \\
\vdots \\
\hat{w}_e(\set_{r-1}) \\
\hat{w}_e(\set_r)
\end{matrix}\right].
\end{equation}
Replace $\hat{w}(\set_i)$ with $w(\set_i)$ and invert the system~\eqref{e:asym_combined_mat} to find the valid coefficients $a_i$.
For convenience, we introduce $b_{r+1}=\gamma_e(e)$.
The inverse of $\mathbf{B}_2$ can be written as 
\begin{equation*}
\frac{1}{b_{r+1}}\cdot
\left[\begin{smallmatrix}
\frac{b_2}{(b_2-b_1)b_1} & -\frac{1}{b_2-b_1} & 0 & \cdots & 0 & 0 & 0 \\
-\frac{1}{b_2-b_1} & \frac{b_3-b_1}{(b_3-b_2)(b_2-b_1)} & -\frac{1}{b_3-b_2} & \cdots & 0 & 0 & 0 \\
\vdots & \vdots & \vdots & \ddots & \vdots & \vdots & \vdots\\
0 & 0 & 0 & \cdots & -\frac{1}{b_{r-1}-b_{r-2}} & \frac{b_r-b_{r-2}}{(b_r-b_{r-1})(b_{r-1}-b_{r-2})} & -\frac{1}{b_r-b_{r-1}} \\
0 & 0 & 0 & \cdots & 0 & -\frac{1}{b_r-b_{r-1}} & \frac{b_{r+1}-b_{r-1}}{(b_{r+1}-b_r)(b_r-b_{r-1})}
\end{smallmatrix}\right].
\end{equation*}
Notice that $\mathbf{B}_2^{-1}$ has the same structure as $\mathbf{B}_1^{-1}$ except for the last element as well as the scaling coefficient  $1/b_{r+1}$.
Since $a_1,\cdots,a_r$ need to be non-negative, we are left to prove the following inequalities:
\begin{align}
w_e(\set_1) &\textstyle\geq \frac{b_1}{b_2} w_e(\set_2), \label{e:ineq_a1}\\
w_e(\set_i) &\textstyle\geq \frac{b_{i+1}-b_i}{b_{i+1}-b_{i-1}}w_e(\set_{i-1}) + \frac{b_i-b_{i-1}}{b_{i+1}-b_{i-1}}w_e(\set_{i+1}), \quad\forall i=2,\cdots,r-1, \label{e:ineq_a2}\\
w_e(\set_r) &\textstyle\geq\frac{b_{r+1}-b_r}{b_{r+1}-b_{r-1}} w_e(\set_{r-1}). \label{e:ineq_a3}
\end{align}
By introducing $b_0=0$, the above inequalities can be rewritten and summarized as
\begin{equation}
g_e(b_i) \textstyle\geq \frac{b_{i+1}-b_{i}}{b_{i+1}-b_{i-1}}g_e(b_{i-1}) + \frac{b_{i}-b_{i-1}}{b_{i+1}-b_{i-1}}g_e(b_{i+1}), \quad\forall i=1,\cdots,r,
\end{equation}
which immediately follows (ii) in Lemma~\ref{lemma:concave}. 
\end{proof}

For the two cases discussed in Theorem~\ref{thm:hg2g}, the required number of building gadgets $|\mathcal{Q}_s|$ and $|\mathcal{Q}_a|$ are respectively upper bounded by $2^{|e|-1}-1$ and $2^{|e|}-2$.
For trivial EDVWs, the theorem coincides with the results for cardinality-based splitting functions in~\citep{veldt2020hypergraph} with $|\mathcal{Q}_s|$ and $|\mathcal{Q}_a|$ respectively reducing to $\lfloor|e|/2\rfloor$ and $|e|-1$.
Hence, EDVWs-based splitting functions generally lead to a denser graph than cardinality-based splitting functions in the worst case.
In addition, when $g_e$ is concave as well as symmetric, Theorem~\ref{thm:hg2g} provides two ways for the reduction while the one leveraging symmetric EDVWs-based gadgets requires fewer edges.

\subsection{Sparsifying hypergraph-to-graph reductions}\label{ss:sparsify}

Graph min-cut/max-flow algorithms have a complexity that depends on the number of vertices and the number of edges in the graph~\citep{goldberg1998recent}. 
The reduction procedures discussed above may result in large and dense graphs, thus affecting the efficiency of algorithms applied to the reduced graph.
A workaround is to find a smaller and sparser graph whose cut approximates, rather than exactly recovers, the hypergraph cut~\citep{bansal2019new, benczur1996approximating, benson2020augmented, chekuri2018minimum, kogan2015sketching}.
In~\citep{benson2020augmented}, a sparsification technique is proposed for hypergraphs with submodular cardinality-based splitting functions.
It is based on approximating concave functions using piecewise linear curves and can be generalized to our EDVWs-based case due to the formulation~\eqref{e:EDVWs_w}.
In the following, we discuss this in detail.

As shown in~\eqref{e:hg_cut}, the hypergraph cut function is defined as the sum of splitting functions associated with every hyperedge.
Hence, the problem can be decomposed into approximately modeling each hyperedge using a sparse gadget with a smaller set of auxiliary vertices.
More formally, we would like to find such a gadget whose corresponding splitting function $\hat{w}_e$ as defined in~\eqref{e:gadget_w} approximates the splitting penalties associated with a hyperedge, i.e.,
\begin{equation}\label{e:epsilon_approximate}
w_e(\set) \leq \hat{w}_e(\set) \leq (1+\epsilon)w_e(\set), \quad\forall\set\subseteq e,
\end{equation}
where $\epsilon\geq 0$ is an approximation tolerance parameter.
This is equivalent to finding some gadget splitting function $\tilde{w}_e$ satisfying $\frac{1}{\delta}w_e(\set)\leq\tilde{w}_e(\set)\leq\delta w_e(\set)$ with the correspondence $\hat{w}_e(\set)=\delta\tilde{w}_e(\set)$ and $\epsilon=\delta^2-1$.

We first consider EDVWs-based splitting functions with concave and symmetric $g_e$.
We have shown in Theorem~\ref{thm:hg2g} that these splitting functions can be exactly constructed from a combination of $|\mathcal{Q}_s|$ symmetric EDVWs-based gadgets. 
Following the idea in~\citep{benson2020augmented}, we next show how to approximate these splitting functions using a smaller set of symmetric EDVWs-based gadgets.
Recall from the proof of Theorem~\ref{thm:hg2g} that the combination of $r$ symmetric EDVWs-based gadgets respectively with positive parameters $b_1<\cdots<b_r$ and combination coefficients $a_1,\cdots,a_r$ has the splitting function in the form of~\eqref{e:sym_combined}.
Its continuous extension can be written as follows.
Since it is symmetric, we can just consider the first half of the function
\begin{equation}\label{e:sym_continuous}
\textstyle \hat{g}_e(x) = \sum_{i=1}^r a_i\cdot\min\{x,b_i\}, \quad\text{where } x\in[0,\gamma_e(e)/2]. 
\end{equation}
When $x \leq b_1$, \eqref{e:sym_continuous} can be rewritten as $\hat{g}_e(x)=(\sum_{i=1}^r a_i) \cdot x$; when $x > b_r$, we have $\hat{g}_e(x)=\sum_{i=1}^r a_ib_i$; when $b_{i-1} < x \leq b_i$ for any $i=2,\cdots,r$, we have $\hat{g}_e(x)=\sum_{j=1}^{i-1}a_jb_j + (\sum_{j=i}^r a_j) \cdot x$.
Hence, $\hat{g}_e(x)$ can also be characterized as the lower envelope of a set of $r+1$ linear functions having non-negative decreasing slopes and non-negative increasing intercepts~\citep{benson2020augmented}, i.e.,
\begin{align}\label{e:piecewise_linear}
& \hat{g}_e(x) = \min \{f_1(x),f_2(x),\cdots,f_{r+1}(x)\}, \\
& \text{where}\quad f_i(x)=m_ix+d_i, \nonumber\\
& \textstyle\hspace{3.5em} m_i=\sum_{j=i}^r a_j \text{ for } 1\leq i\leq r \text{ and } m_{r+1}=0, \nonumber\\
& \textstyle\hspace{3.5em} d_1=0 \text{ and } d_i=\sum_{j=1}^{i-1}a_jb_j \text{ for } 2\leq i\leq r+1.  \nonumber
\end{align}
Equivalently, the relations between the coefficients $a_i,b_i$ and $m_i,d_i$ can also be described as $a_i=m_i-m_{i+1}$ and $b_i=(d_{i+1}-d_i)/a_i$.

The sparsification problem can be described as: Find the piecewise linear function $\hat{g}_e$ with the minimum number of pieces that approximates $g_e$ at the points in $\mathcal{Q}_s$.
It can be formulated as follows.
\begin{align}\label{e:sparsification}
& \min\,\, r \\
& \mathrm{s.t.}\quad g_e(x)\leq\hat{g}_e(x)\leq(1+\epsilon)g_e(x), \forall x\in\mathcal{Q}_s, \nonumber\\
& \hspace{2.4em} \hat{g}_e(x) \text{ is defined as}~\eqref{e:piecewise_linear}, \nonumber\\
& \hspace{2.4em} \text{For each } i\in[r+1], f_i(x)=g_e(x) \text{ for some } x\in\{0\}\cup\mathcal{Q}_s. \nonumber
\end{align}
The first constraint is from~\eqref{e:epsilon_approximate}.
The third constraint is added without loss of generality since an improved approximation could be found if there is some $f_i$ strictly greater than $g_e$ at all points in $\{0\}\cup\mathcal{Q}_s$.
The main difference with the corresponding formulation in~\citep{benson2020augmented} is that, for the cardinality-based case there considered, the set $\mathcal{Q}_s$ consists of only integers from $1$ to $\lfloor|e|/2\rfloor$; for the EDVWs-based case, the elements in $\mathcal{Q}_s$ depend on the values of EDVWs thus they are not necessarily integers or evenly spaced.

We can modify the algorithm proposed in~\citep{benson2020augmented} to find an optimal solution to~\eqref{e:sparsification}. 
Here we briefly introduce the procedure and please refer to~\citep{benson2020augmented} for a detailed explanation.
For convenience, we denote the elements in $\mathcal{Q}_s$ by $q_1<q_2<\cdots<q_n$ where $n=|\mathcal{Q}_s|$ and define a series of functions with the $i$th one joining $(q_i,g_e(q_i))$ and $(q_{i+1},g_e(q_{i+1}))$, i.e., $h_i(x)=\frac{g_e(q_{i+1})-g_e(q_i)}{q_{i+1}-q_i}(x-q_i)+g_e(q_i)$.
The last linear piece $f_{r+1}$ is first determined. 
According to~\eqref{e:piecewise_linear} and~\eqref{e:sparsification}, it has a zero slope and passes through $(q_n, g_e(q_n))$, thus we have $f_{r+1}(x)=g_e(q_n)$.
For the first linear piece $f_1$, it goes through the origin according to~\eqref{e:piecewise_linear} and its slope $m_1$ is set to $\frac{g_e(q_1)}{q_1}$ so that it can provide a qualified approximation for as many points in $\mathcal{Q}_s$ as possible.
Then we identify the first point $q_\ell$ in $\mathcal{Q}_s$ that does not have a $(1+\epsilon)$-approximation.
In other words, $f_1(q_i)\leq(1+\epsilon)g_e(q_i)$ holds for $i=1,\cdots,\ell-1$ and becomes invalid since $i=\ell$. 
We stop searching more linear pieces if (i) $\ell\geq n$, or (ii) $g_e(q_n)\leq(1+\epsilon)g_e(q_\ell)$ where (ii) implies that $q_\ell,\cdots,q_n$ have a $(1+\epsilon)$-approximation provided by the last linear piece.
Otherwise, we continue to find the next linear piece $f_2$ in order to cover the rest of the points.
We identify $i^{*}$ such that $h_{i^{*}-1}(q_\ell)\leq(1+\epsilon)g_e(q_\ell)$ and $h_{i^{*}}(q_\ell)>(1+\epsilon)g_e(q_\ell)$.
In other words, the linear function $h_{i^{*}-1}$ provides a qualified approximation at $q_\ell$ (i.e., the first point has not been covered yet), but the following functions $h_i$ for $i \geq i^{*}$ do not.
Hence, $f_2$ should pass through the point $(q_{i^{*}},g_e(q_{i^{*}}))$ and lie between $h_{i^{*}-1}$ and $h_{i^{*}}$ (its slope should be greater than $h_{i^{*}}$'s slope and no greater than $h_{i^{*}-1}$'s slope).
In the meantime, we expect $f_2$ to provide a qualified approximation at $q_\ell$ and have a slope as small as possible so that more points after $q_{i^{*}}$ can be covered.
Therefore, we set $f_2$ to the line joining $(q_\ell,(1+\epsilon)g_e(q_\ell))$ and $(q_{i^{*}},g_e(q_{i^{*}}))$.
We refer the reader to Figure~4 in \citep{benson2020augmented} for an illustration.
Then we update $q_\ell$ to the new start point in $\mathcal{Q}_s$ that does not have a $(1+\epsilon)$-approximation yet, and check the stopping criterion.
We repeat the process of picking $f_2$ to add more linear pieces until we meet the stopping criterion.

\begin{figure*}[t!]
\centering
\includegraphics[scale=0.5]{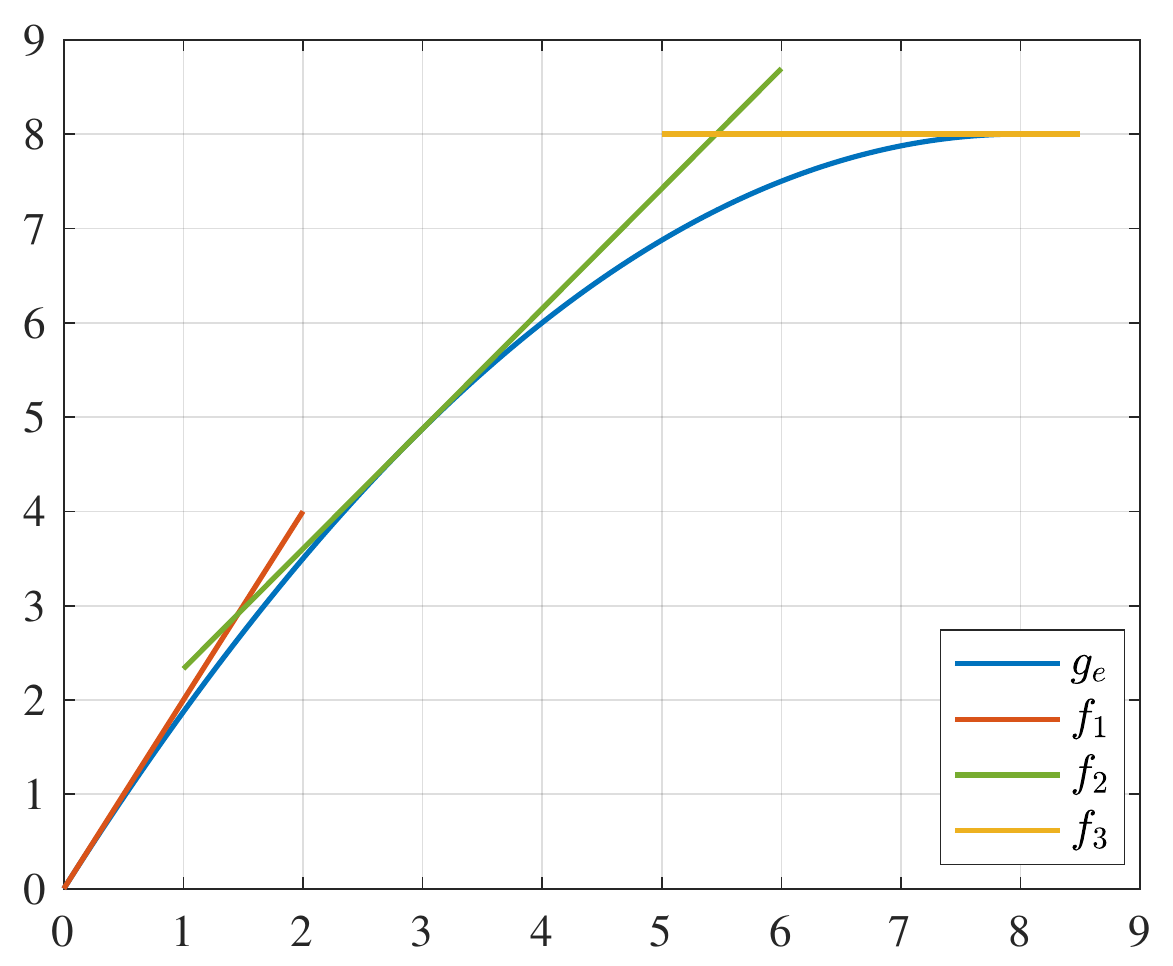}
\captionsetup{margin=0.3cm}
\caption{An example where $g_e(x)=-0.125x^2+2x$. We want to find a $(1+\epsilon)$-approximation for $g_e(x)$ everywhere in the range $[0,8]$ where we set $\epsilon=0.1$. The last linear piece $f_3$ has a zero slope and passes through $(8,8)$, thus $f_3(x)=8$. The first linear piece $f_1$ is tangent to $g_e$ at the origin, hence $f_1(x)=2x$. At the point $x=1.4545$, $f_1(x)=(1+\epsilon)g_e(x)$, meaning that $f_1$ provides a qualified approximation for $g_e(x)$ in the range $[0,1.4545]$. The second linear piece $f_2$ passes through $(1.4545, 2.909)$ and is tangent to $g_e$ at $(2.909, 4.7602)$, namely $f_2(x)=1.2727x+1.0579$. At $x=5.2894$, $f_2(x)=(1+\epsilon)g_e(x)$, hence $f_2$ provides a qualified approximation for $g_e(x)$ in the range $[1.4545,5.2894]$. The rest of the points in the range $[5.2894,8]$ have been covered by $f_3$.}
\label{fig:example}	
\end{figure*}

We can also ignore the particular positions of elements in $\mathcal{Q}_s$ and try to provide a $(1+\epsilon)$-approximation for $g_e(x)$ everywhere in the range $[0,\gamma_e(e)/2]$.
This approach has two benefits: 
(i) It avoids building the set $\mathcal{Q}_s$. 
(ii) If multiple hyperedges share the same continuous extension $g_e$, we can find their sparsified reductions all at once.
The procedure of finding the set of linear pieces can be modified as follows.
The last linear piece should be $f_{r+1}(x)=g_e(\gamma_e(e)/2)$. 
The first linear piece $f_1$ is tangent to $g_e$ at the origin.
We identify the value $z\in[0,\gamma_e(e)/2]$ that satisfies $f_1(z)=(1+\epsilon)g_e(z)$.
If no such a $z$ exits or $g_e(\gamma_e(e)/2)\leq(1+\epsilon)g_e(z)$, we stop.
Otherwise, we add the next linear piece $f_2$ which is selected to pass through $(z,(1+\epsilon)g_e(z))$ and be tangent to $g_e$ as some point greater than $z$.
We repeat the process of choosing $f_2$ to add more linear pieces until the whole range $[0,\gamma_e(e)/2]$ has been covered.
An illustrative example is given in Figure~\ref{fig:example}. 
 
For EDVWs-based splitting functions with concave and possibly asymmetric $g_e$, they can be approximated using a smaller set of asymmetric EDVWs-based gadgets.
It has been proved in~\citep{benson2020augmented} that the continuous extension of the splitting function in the form of~\eqref{e:asym_combined} is also piecewise linear, thus we can adopt a similar reduction procedure as described above.

Moreover, a direct extension of Theorem 4.1 in~\citep{benson2020augmented} is that the number of symmetric/asymmetric EDVWs-based gadgets needed to approximate an EDVWs-based splitting function with concave, symmetric/asymmetric $g_e$ is upper bounded by $O(\log_{1+\epsilon}\gamma_e(e))$ which behaves as $\epsilon^{-1}\log\gamma_e(e)$ as $\epsilon$ approaches zero.
This bound could be further improved if a specific concave function $g_e$ is chosen. 

\section{Experiments}\label{s:exp}

\begin{figure*}[t!]
\centering
\includegraphics[scale=0.35]{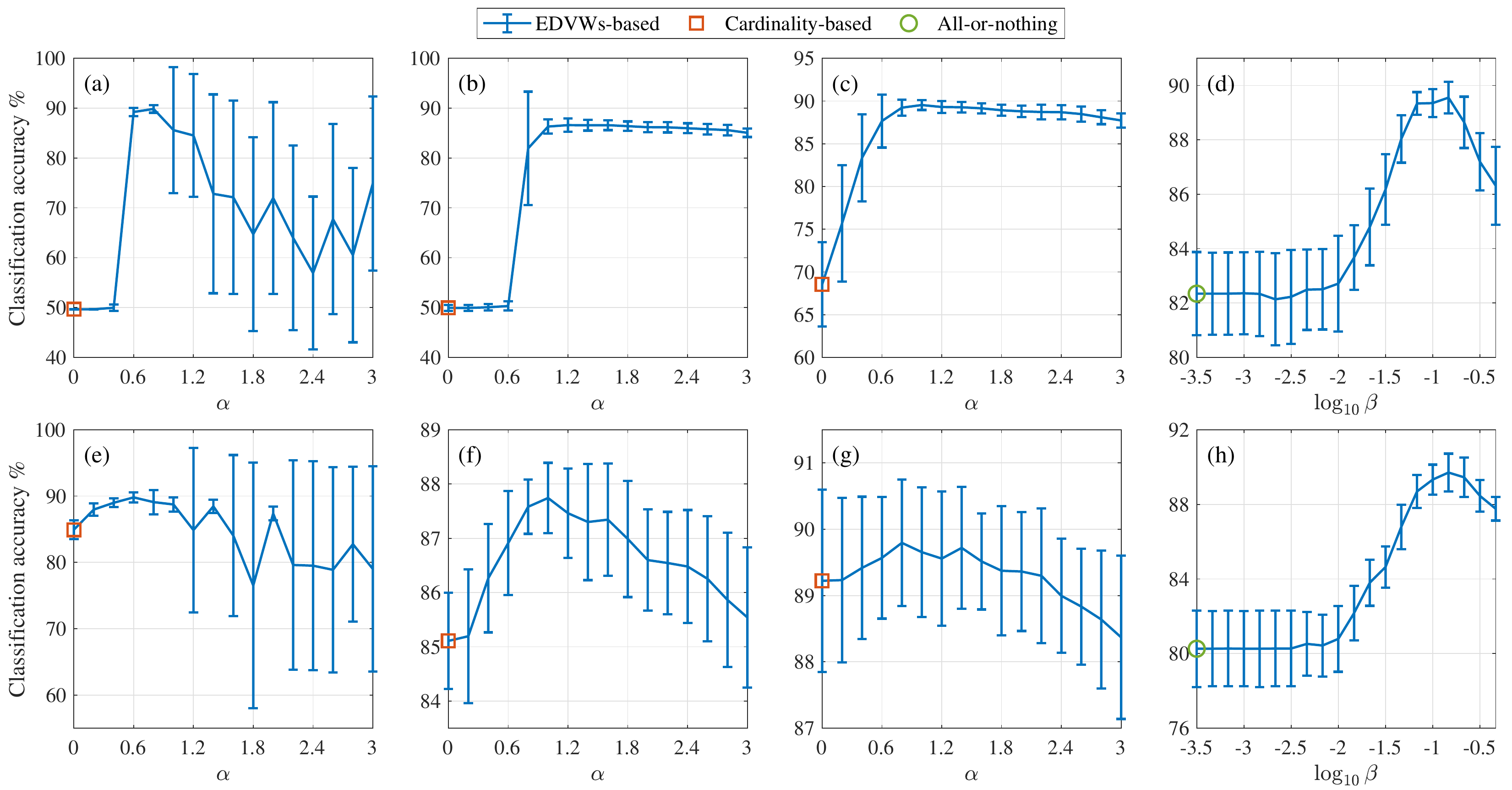}
\captionsetup{margin=0.3cm}
\caption{Classification performance as a function of the parameters $\alpha$ and $\beta$. For the two rows, the fraction of labeled vertices is respectively set to $0.3$ and $0.5$. (a) and (e) correspond to the splitting function $w_e(\set)=\gamma_e(\set)\cdot\gamma_e(e\setminus\set)$; (b) and (f) correspond to the splitting function $w_e(\set)=\min\{\gamma_e(\set),\gamma_e(e\setminus\set)\}$; (c-d) and (g-h) correspond to the splitting function $w_e(\set)=\min\{\gamma_e(\set),\gamma_e(e\setminus\set),\beta\gamma_e(e)\}$ where we fix $\beta=0.15$ in (c), (g) and we fix $\alpha=1$ in (d), (h).}
\label{fig:res1}	
\end{figure*}

We show the effects of introducing EDVWs into hypergraph cuts via numerical experiments.
We consider the binary classification of hypergraph vertices: Given the labels of a subset of vertices $\V^L=\V^L_1\cup\V^L_2$ where $\V^L_1$ and $\V^L_2$ respectively consist of all labeled vertices in two classes, the task is to estimate the labels of the rest of the vertices.
The problem can be formulated as a generalized hypergraph minimum $s$-$t$ cut problem with multiple source and sink vertices, i.e.,
\begin{equation}\label{e:hg_st_m}
\textstyle \min_{\emptyset\subset\set\subset\V} \, \cut_{\hg}(\set) \st \V^L_1\subseteq\set, \V^L_2\subseteq\V\setminus\set.
\end{equation}
Then the vertices in $\set$ and $\V\setminus\set$ are classified into two categories, respectively.
We consider the first three EDVWs-based splitting functions listed in Table~\ref{table:sf}, namely $w_e(\set)=\gamma_e(\set)\cdot\gamma_e(e\setminus\set)$, $w_e(\set)=\min\{\gamma_e(\set),\gamma_e(e\setminus\set)\}$ and $w_e(\set)=\min\{\gamma_e(\set),\gamma_e(e\setminus\set),b\}$, which are symmetric and graph reducible.
For any of them, we can reduce the hypergraph to a graph sharing the same cut properties.
Following the idea in~\citep{blum2001learning}, we introduce another two vertices -- a super-source $s$ and a super-sink $t$ -- into the reduced graph.
For every $v\in\V^L_1$, add an edge of weight infinity from $s$ to $v$; for every $v\in\V^L_2$, add an edge of weight infinity from $v$ to $t$.
Then problem~\eqref{e:hg_st_m} can be converted into a common minimum $s$-$t$ cut problem defined on the new graph in the form of~\eqref{e:g_st}. 
We solve the graph minimum $s$-$t$ cut problem using the highest-label preflow-push algorithm implemented in the NetworkX package~\citep{hagberg2008exploring}.

We adopt the $20$ Newsgroups dataset and consider the task of document classification.
The dataset contains documents in different categories and we consider the documents in categories ``rec.motorcycles" and ``sci.space".
We extract the $200$ most frequent words in the corpus after removing stop words and words appearing in $>3\%$ and $<0.2\%$ of the documents.
We then remove a small fraction of documents that do not contain the selected words and finally get $1852$ documents with $932$ and $920$ documents in two classes, respectively.
To model the text dataset using hypergraphs with EDVWs, we consider documents as vertices and words as hyperedges.
A document (vertex) belongs to a word (hyperedge) if the word appears in the document.
The EDVWs are taken as the corresponding tf-idf (term frequency-inverse document frequency) values~\citep{leskovec2020mining} to the power of $\alpha$, where $\alpha$ is a tunable parameter.
More precisely, we set
\begin{equation}
\gamma_e(v) = \text{tf-idf}(e,v)^\alpha, \quad \text{where } \text{tf-idf}(e,v) = \text{tf}(e,v) \cdot \text{idf}(e).
\end{equation}
The term frequency $\text{tf}(e,v)$ is the relative frequency of word $e$ in document $v$.
The inverse document frequency $\text{idf}(e)$ measures the informativeness of word $e$, i.e., if it is common or rare across all documents.
Hence, the tf-idf values are able to reflect the importance of a word to a document in a corpus and thus an ideal choice for EDVWs.
We adopt the TfidfTransformer function in the scikit-learn package with default parameters to compute the tf-idf values.
The parameter $\alpha$ is introduced for extra flexibility.
When $\alpha=0$, we get the trivial EDVWs and the splitting functions reduce to cardinality-based ones.
For the splitting function $w_e(\set)=\min\{\gamma_e(\set),\gamma_e(e\setminus\set),b\}$, we set $b=\beta\gamma_e(e)$ where $\beta$ is also adjustable.
If a small enough $\beta$ is selected, the splitting function reduces to the all-or-nothing case.

\begin{figure*}[t!]
\centering
\includegraphics[scale=0.35]{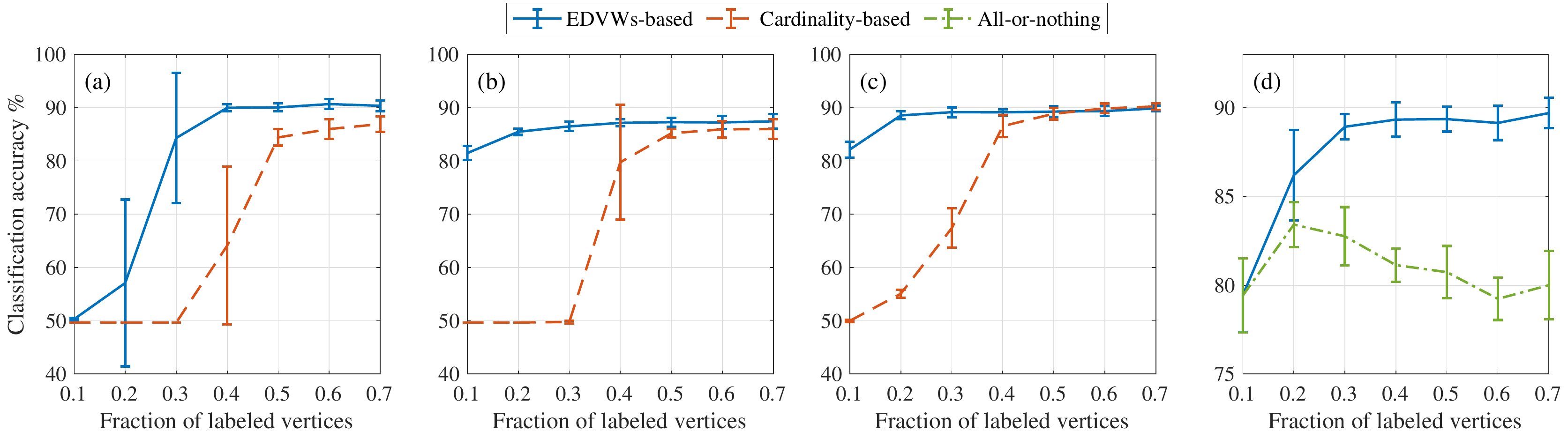}
\captionsetup{margin=0.3cm}
\caption{Performance comparison between the proposed splitting functions and existing ones. (a), (b) and (c-d) respectively correspond to the splitting functions $w_e(\set)=\gamma_e(\set)\cdot\gamma_e(e\setminus\set)$, $w_e(\set)=\min\{\gamma_e(\set),\gamma_e(e\setminus\set)\}$ and $w_e(\set)=\min\{\gamma_e(\set),\gamma_e(e\setminus\set),\beta\gamma_e(e)\}$. We fix $\beta=0.15$ in (c) and fix $\alpha=1$ in (d).
The red curves (cardinality-based) and the green curve (all-or-nothing) respectively correspond to the cases when $\alpha=0$ and when $\beta$ is small enough ($\beta=10^{-3.5}$ here).
For the blue curves (EDVWs-based), a $5$-fold cross-validation is adopted in (a-c) to search the optimal $\alpha$ and in (d) to search the optimal $\beta$.}
\label{fig:res2}	
\end{figure*} 

Figure~\ref{fig:res1} shows the effects of the parameters $\alpha$ and $\beta$ on the classification performance.
We plot the average classification accuracy and the standard deviation over $10$ realizations which adopt different sets of labeled vertices.
We respectively set the fraction of labeled vertices to $0.3$ in (a-d) and $0.5$ in (e-h).
The three considered EDVWs-based splitting functions respectively correspond to (a) (e), (b) (f), and (c-d) (g-h).
For the third splitting function, we fix $\beta=0.15$ to observe the influence of $\alpha$ in (c) (g) and fix $\alpha=1$ to test $\beta$ in (d) (h).
It can be observed that, for all of them, the best performance is achieved for intermediate values of $\alpha$ or $\beta$ rather than the extreme cases when the EDVWs-based splitting function reduces to a cardinality-based splitting function or the all-or-nothing splitting function.

Figure~\ref{fig:res2} provides a more direct comparison between the proposed EDVWs-based splitting functions and existing ones.
For EDVWs-based splitting functions, we adopt a $5$-fold cross-validation to find the optimal $\alpha$ and $\beta$.
In (a), we search the optimal $\alpha$ over the set $\{0:0.2:3\}$.
In (b-c), we search $\alpha$ over $\{0:0.2:5\}$.
In (d), we search $\beta$ over $20$ equally spaced values between $10^{-3.5}$ and $10^{-1/3}$ in the log scale.
We can see that adopting EDVWs-based splitting functions improves the classification performance over a wide range of train-test split ratios.

\section{Conclusion}\label{s:end}

We developed a framework for incorporating EDVWs into hypergraph cut problems and generalized reduction as well as sparsification techniques recently proposed for cardinality-based splitting functions.
Through a real-world text mining application, we showcased the value of the introduction of EDVWs.
There are numerous directions for future work:
(i) As mentioned in the introduction, hypergraph minimum cuts can be used to solve various real-world applications~\citep{catalyurek1999hypergraph, ding2008image, karypis1999multilevel, kim2011higher} or as subroutines in many machine learning algorithms~\citep{liu2021strongly, veldt2020minimizing}.
Hence, it would be desirable to apply the proposed framework to these applications and algorithms and evaluate the performance.
(ii) Another direction is to further extend the proposed framework to multiway cuts~\citep{chekuri2011approximation, okumoto2012divide, veldt2020hypergraph, zhao2005greedy} or other types of cuts such as normalized cuts~\citep{li2017inhomogeneous, li2018submodular, fountoulakis2021local}.
(iii) An open problem is whether all submodular splitting functions are graph reducible~\citep{veldt2020hypergraph}.

\section*{\small Declarations}

\begin{backmatter}

\section*{Abbreviations}
EDVWs: Edge-dependent vertex weights; 
VLSI: Very large scale integration;
tf-idf: Term frequency-inverse document frequency

\section*{Acknowledgements}
Not applicable

\section*{Authors' contributions}
YZ developed the methodology, conducted the experiments, and wrote the initial manuscript.
SS supervised the study, contributed to the discussions, and edited the manuscript.
Both authors read and approved the final manuscript.

\section*{Funding}
This work was supported by NSF under award CCF 2008555.

\section*{Availability of data and materials}
The code and data are available at~\url{https://github.com/yuzhu2019/hg_cut_edvws}. 

\section*{Competing interests}
The authors declare that they have no competing interests.

\section*{Consent for publication}
Not applicable

\section*{Ethics approval and consent to participate}
Not applicable

\bibliographystyle{spbasic} 
\bibliography{references.bib}        
\nocite{label}

\end{backmatter}
\end{document}